\documentclass[12pt, draftclsnofoot, onecolumn]{IEEEtran}
\usepackage{subfigure,cite,graphicx,psfig,amsmath,amssymb,eufrak,mathrsfs,epsf,epsfig}
\usepackage[usenames]{color}

\usepackage[active]{srcltx}
\ifCLASSINFOpdf

\else
\fi

\begin{document}
\newtheorem{ach}{Achievability}
\newtheorem{con}{Converse}
\newtheorem{definition}{Definition}
\newtheorem{theorem}{Theorem}
\newtheorem{lemma}{Lemma}
\newtheorem{example}{Example}
\newtheorem{cor}{Corollary}
\newtheorem{prop}{Proposition}
\newtheorem{conjecture}{Conjecture}
\newtheorem{remark}{Remark}
\title{Degrees of Freedom of Interference Channels with Hybrid Beam-forming}
\author{\IEEEauthorblockN{Sung Ho Chae,~\IEEEmembership{Member,~IEEE} and Cheol Jeong,~\IEEEmembership{Member,~IEEE}
}\thanks{The authors are with the DMC R\&D Center, Samsung Electronics, Suwon, Republic of Korea.}%
}
 \maketitle

\begin{abstract}
We study the sum degrees of freedom (DoF) of interference channels with hybrid beam-forming in which each transmitter $i$ uses $M_i'$ antennas and $M_i$ RF chains and each receiver $i$ uses $N_i'$ antennas and $N_i$ RF chains, where $M_i\leq M_i'$ and $N_i\leq N_i'$, $\forall i=1,2,\ldots,K$, and hybrid beam-forming composed of analog and digital precodings is employed at each node. For the two-user case, we completely characterize the sum DoF for an arbitrary number of antennas and RF chains by developing an achievable scheme optimized for the hybrid beam-forming structure and deriving its matching upper bound. For a general $K$-user case, we focus on a symmetric case where $M_i=M$, $N_i=N$, $M_i'=M'$, and $N'_i=N'$, $\forall i=1,2,\ldots,K$, and obtain lower and upper bounds on the sum DoF, which are tight when $\frac{\max\{M',N'\}}{\min\{M',N'\}}$ is an integer. The results show that hybrid beam-forming can increase the sum DoF of interference channel under certain conditions while it cannot improve the sum DoFs of point-to-point channel, multiple access channel, and broadcast channel. The key insights on this gain is that hybrid beam-forming enables users to manage inter-user interference better, and thus each user can increase the dimension of interference-free signal space for its own desired signals.
\end{abstract}
 \IEEEpeerreviewmaketitle



\section{Introduction}
Mobile data traffic has been growing dramatically as the number of mobile smart devices is increasing rapidly in recent years~\cite{CISCO13}. To accommodate tremendous demand on mobile data traffic, the cell capacity can be largely increased by deploying a very large number of antennas at base stations (BSs), often referred to as a massive multiple-input multiple-output (MIMO) system~\cite{Marzetta:TWC10,LuLiSwindlehurstAshikhminZhang:JSTSP14}. The massive MIMO system, however, has hardware constrains that come from using a few hundred antennas. {For a conventional antenna array structure}, each antenna needs to have a dedicated RF chain. This naturally leads to an increment in the circuit size, power consumption, and device cost proportionally to the number of antennas, and hence it can be a serious problem in a practical point of view especially for massive MIMO systems. {Therefore, to resolve this problem}, a hybrid beam-forming structure with a lower number of RF chains than the number of antenna elements {has been recently introduced as a practical solution~\cite{ZhangMolischKung:TSP05,VenkateswaranVeen:TSP10}.}

As an alternative approach to increase the cell capacity, millimeter-wave (mmWave) communications have attracted great attention recently~\cite{RappaportSunMayzusZhaoAzarWangWongSchulzSamimiGutierrez:Access13}. The mmWave band from 30 to 300 GHz provides abundant contiguous frequency resources while frequency bands under 5 GHz used for legacy cellular communications are very crowded and fragmented. {The main advantage in mmWave communications is that} a very high data rate can be supported using a very large bandwidth at mmWave bands. {However}, one of major drawbacks is the high induced path loss due to the propagation loss and absorption loss at mmWave bands~\cite{MarcusPattan:MM05}. {Fortunately}, this high path loss can be effectively compensated by a high beam-forming gain obtained from a large number of antenna elements that can be packed into a small form factor due to the small wavelength in mmWave bands. {To support a single stream only, the analog beam-forming, which is simply implemented by controlling attenuators and phase shifters of the antenna array to steer a directional beam, is enough to be considered. However, to transmit multiple streams}, the hybrid beam-forming structure, where analog beam-forming is performed at RF domain and antenna arrays are connected to a relatively small number of digital paths, {should be} considered to get the multiplexing gain~\cite{ElAyachRajagopalAbu-SurraPiHeath:TWC14,AlkhateebMoGonzalez-PrelcicHeath:CM14}.

{As mentioned above}, the hybrid beam-forming architecture can play a key role in the next generation communications (e.g., massive MIMO and/or mmWave communications) {and hence has been widely studied recently}~\cite{ElAyachRajagopalAbu-SurraPiHeath:TWC14,AlkhateebMoGonzalez-PrelcicHeath:CM14,WuChiuLinChang:TWC13,AlkhateebElAyachLeusHeath:JSTP14,LiangXuDong:WCL14}. In~\cite{ElAyachRajagopalAbu-SurraPiHeath:TWC14}, precoders and combiners are designed using a sparse reconstruction approach. In~\cite{WuChiuLinChang:TWC13}, baseband and RF beams are designed for multiuser downlink spatial division multiple access (SDMA). In addition, a hybrid precoding algorithm based on a hierarchical codebook is proposed in~\cite{AlkhateebElAyachLeusHeath:JSTP14}. Furthermore, a hybrid precoder is proposed for massive multiuser MIMO systems in~\cite{LiangXuDong:WCL14}. While there are some works on hybrid beam-forming structures, however, to the best of our knowledge, the degrees of freedom (DoF) gain from hybrid beam-forming has not been analyzed before.

\subsection{Previous Works}
The DoF, which is also known as a capacity pre-log, gives the capacity approximation at high signal to noise ratio (SNR) regime. For example, for the point-to-point (PTP) channel with $M$ transmit antennas and $N$ receive antennas, it is well known that the capacity increases with the growth rate $\min\{M,N\}\log(\text{SNR})$ at high SNR~\cite{Foschini98,Telatar99}. Since exact capacity characterization is generally still unknown even for simple networks (e.g., two-user interference channel), instead of obtaining an exact capacity, approximate characterization by finding the optimal DoF has been studied in many networks recently~\cite{Jafar07,Cadambe107,Viveck1:09,Viveck2:09,Motahari:091,Suh08,Suh:11,Akbar10,Tiangao:10,Jeon2:11,Tiangao:12,Annapureddy:11,Ke:12,Jeon4:12,ChaeJ1,Chae11,Krishnamurthy12}.

Specifically, for the two-user interference channel, the sum DoF has been completely characterized, where zero-forcing precoding has been shown to be enough to achieve the optimal DoF~\cite{Jafar07}. For a general $K$-user interference channel, a novel interference management technique called \emph{interference alignment} has been proposed in~\cite{Cadambe107,Motahari:091}, which achieves the optimal sum DoF of $\frac{K}{2}$. Later this scheme has been extended to MIMO configurations both for rich scattering environment~\cite{Tiangao:10,Akbar10} and poor scattering  environment~\cite{ChaeJ1,Chae11,Krishnamurthy12}. Furthermore, beyond the interference channels, the idea of interference alignment has been successfully adapted to various networks, e.g., see~\cite{Viveck1:09,Viveck2:09,Motahari:091,Suh08,Suh:11,Jeon2:11,Tiangao:12,Annapureddy:11,Ke:12,Jeon4:12} and references therein.

\subsection{Contributions}
In this paper, our primary goal is to answer if hybrid beam-forming can increase the sum DoF of interference channels. To this end, motivated by the aforementioned previous works, we propose zero forcing and interference alignment schemes optimized for the hybrid beam-forming structure. In addition, we also derive a new upper bound on the sum DoF when hybrid beam-forming is employed at each node. For the two-user case, this upper bound coincides with the achievable sum DoF of the proposed scheme, thereby completely characterizing the sum DoF. For a general $K$-user case, our proposed scheme can achieve the upper bound when $\frac{\max\{M',N'\}}{\min\{M',N'\}}$ is an integer, where $M'$ and $N'$ denote the number of antennas at each transmitter and receiver, respectively. As a consequence of the result, we show that hybrid beam-forming can indeed improve the sum DoF of the $K$-user interference channel under certain conditions. This is in contrast to the PTP channels, multiple access channel (MAC), and broadcast channel (BC) cases in which hybrid beam-forming cannot increase the sum DoF (see Section III). The key insight behind this gain is that hybrid beam-forming enables users to manage interference better, and thus each user can increase the dimension of interference-free signal space which can be used for its own desired signals.

\subsection{Organization}
The rest of this paper is organized as follows. In Section II, we describe the system model and sum DoF metric considered in this paper. In Section III, we give an intuition as to how hybrid beam-forming can increase the sum DoF through motivating examples. In Section IV, we present and discuss about the main results of this paper. In addition, we provide numerical results which show the performance improvement from hybrid beam-forming in Section V. In Sections VI and VII, we provide the proofs of the main theorems. Finally, we conclude the paper in Section VIII.

\subsection{Notations}
Throughout the paper, we will use $\mathbf{A}$, $\mathbf{a}$, and $a$ to denote a matrix, a vector, and a scalar, respectively. For a rational number $a$, the notation $\lfloor a \rfloor$ denotes the integer part of $a$. For matrix $\mathbf{A}$, let $\mathbf{A}^{T}$, $\mathbf{A}^{\ast}$, and $||\mathbf{A}||$ denote the transpose, the complex conjugate transpose, and the norm of $\mathbf{A}$, respectively. In addition, let
$|\mathbf{A}|$ and $\text{rank}(\mathbf{A})$ denote the
determinant and the rank of $\mathbf{A}$, respectively. The
notations $\mathbf{I}_n$ and $\mathbf{0}_{n\times n}$ denote
the $n\times n$ identity matrix and zero matrix, respectively. We write $f(x)=o(x)$ if $\lim_{x\rightarrow
\infty}\frac{f(x)}{x}=0$.

\section{System Model}
Consider a $K$-user $(M_i, M_i')\times (N_i, N_i')$ interference channel with hybrid beam-forming, as shown in Fig.~\ref{channel model}.
Transmitter $i$ wishes to communicate with receiver $i$ only, while causing interference to all the other receivers. In addition, transmitter $i$ uses $M_i'$ antennas and $M_i$ RF chains and receiver $i$ uses $N_i'$ antennas and $N_i$ RF chains, where $M_i\leq M_i'$ and $N_i\leq N_i'$, $\forall i=1,2,\ldots,K$. Specially, when $M_i'=M_i$ and $N_i'=N_i$, $\forall i=1,2,\ldots,K$, we call the corresponding channel as a \emph{full digital} channel.

\begin{figure}[t!]
\begin{center}
\includegraphics[scale=0.43]{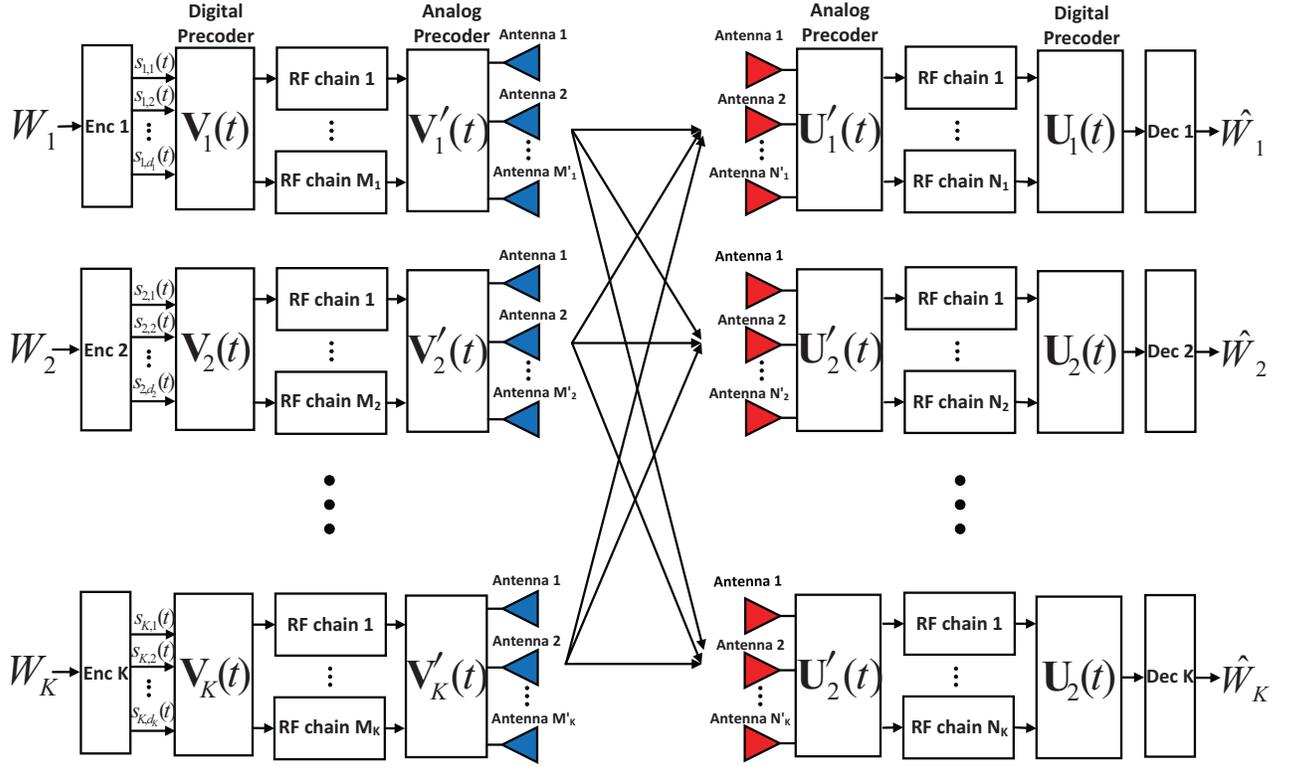}
\end{center}
\vspace{-0.1in}
\caption{The $K$-user interference channel with hybrid beam-forming}
\label{channel model}
\vspace{-0.1in}
\end{figure}

\subsection{Channel Model}
Similar to previous works~\cite{ElAyachRajagopalAbu-SurraPiHeath:TWC14,WuChiuLinChang:TWC13}, in this paper we assume that transmitter $i$ utilizes transmit hybrid beam-forming which consists of an $M_i'\times M_i$ analog precoder $\mathbf{V}_i'(t)$ and an $M_i\times d_i$ digital precoder $\mathbf{V}_i(t)$ as depicted in Fig.~\ref{channel model}, where $d_i\leq \{M_i,N_i\}$ denotes the number of streams of user $i$.\footnote{As compared to the hybrid beam-forming structure introduced in~\cite{ElAyachRajagopalAbu-SurraPiHeath:TWC14,WuChiuLinChang:TWC13}, in this paper, coefficients in $\mathbf{V}'_i(t)$ can have different norms by relaxing the constraint that all entries are of equal norm. In practical point of view, this is feasible since we can implement $\mathbf{V}'_i(t)$ by using both attenuators and analog phase shifters rather than using analog phase shifters only.} In addition, based on this hybrid beam-forming, the input signal of transmitter $i$ at time slot $t$, $\mathbf{x}_i(t)$, is assumed to be given by
\begin{align*}
\mathbf{x}_i(t)&=\mathbf{V}_i'(t)\mathbf{x}^{[b]}_i(t)\\
&=\mathbf{V}_i'(t)\mathbf{V}_i(t)\mathbf{s}_i(t),
\end{align*}
where $\mathbf{x}^{[b]}_i(t)=\mathbf{V}_i(t)\mathbf{s}_i(t)$ is the $M_i\times 1$ baseband-domain input vector and $$\mathbf{s}_i(t)=\left[\begin{array}{cccc}s_{i,1}(t)&\cdots&s_{i,d_i}(t)\end{array}\right]^T$$ is the $d_i\times 1$ symbol vector of transmitter $i$. Here, $s_{i,j}(t)$ denotes the $j$th symbol of user $i$ at time slot $t$. Then the input and output relationship at RF domain is given by
\begin{align*}
\mathbf{y}_{j}(t)&=\sum^{K}_{i=1}\mathbf{H}_{ji}(t)\mathbf{x}_{i}(t)+\mathbf{z}_{j}(t)\\
&=\sum^{K}_{i=1}\mathbf{H}_{ji}(t)\mathbf{V}_i'(t)\mathbf{x}^{[b]}_i(t)+\mathbf{z}_{j}(t)\\
&=\sum^{K}_{i=1}\mathbf{H}_{ji}(t)\mathbf{V}_i'(t)\mathbf{V}_i(t)\mathbf{s}_i(t)+\mathbf{z}_{j}(t),
\end{align*}
where $\mathbf{H}_{ji}(t)$ is the $N'_j\times M'_i$
channel matrix from transmitter $i$ to receiver $j$,
$\mathbf{y}_j(t)$ is the $N'_j\times 1$ RF-domain received signal vector at
receiver $j$, and $\mathbf{z}_j(t)$ is the Gaussian noise vector at receiver $j$ whose entries are drawn from $\mathcal{CN}(0,1)$. We assume that all channel coefficients are independent and identically distributed (i.i.d) from a continuous distribution and known to all nodes.

After receiving $\mathbf{y}_{j}(t)$, receiver $j$ applies receive hybrid beam-forming which consists of an analog precoder $\mathbf{U}_j'(t)$ and a digital precoder $\mathbf{U}_j(t)$ as shown in Fig.~\ref{channel model}. Specifically, by applying the analog precoder to the received signal at RF domain, we can obtain the input and output relationship at baseband domain as
\begin{align*}
\mathbf{y}_j^{[b]}(t)=\sum^{K}_{i=1}\mathbf{H}_{ji}^{[b]}(t) \mathbf{x}_i^{[b]}(t)+\mathbf{z}^{[b]}_j(t),
\end{align*}
where $\mathbf{y}_j^{[b]}(t)=\mathbf{U}'^{\ast}_j(t)\mathbf{y}_j(t)$, $\mathbf{H}^{[b]}_{ji}(t)=\mathbf{U}'^{\ast}_j(t)\mathbf{H}_{ji}(t)\mathbf{V}_i'(t)$, and $\mathbf{z}_j^{[b]}(t)=\mathbf{U}_j'^{\ast}(t)\mathbf{z}_j(t)$. If we further apply the digital precoder to the received signal at baseband domain, we finally get
\begin{align*}
\mathbf{y}_j^{[e]}(t)=\sum^{K}_{i=1}\mathbf{H}_{ji}^{[e]}(t) \mathbf{s}_i(t)+\mathbf{z}^{[e]}_j(t),
\end{align*}
where $\mathbf{y}_j^{[e]}(t)=\mathbf{U}^{\ast}_j(t)\mathbf{y}^{[b]}_j(t)=\mathbf{U}^{\ast}_j(t)\mathbf{U}'^{\ast}_j(t)\mathbf{y}_j(t)$, $\mathbf{H}^{[e]}_{ji}(t)=\mathbf{U}^{\ast}_j(t)\mathbf{U}'^{\ast}_j(t)\mathbf{H}_{ji}\mathbf{V}_i'(t)\mathbf{V}_i(t)$, and $\mathbf{z}_j^{[e]}(t)=\mathbf{U}^{\ast}_j(t)\mathbf{U}_j'^{\ast}(t)\mathbf{z}_j(t)$. Note that $\mathbf{H}^{[e]}_{ji}(t)$ is the effective channel matrix which can be obtained after applying transmit hybrid beam-forming of transmitter $i$ and receive hybrid beam-forming of receiver $j$.

Finally, by applying the aforementioned hybrid beam-forming strategy and assuming Gaussian signaling $\mathbf{s}_i(t)\sim \mathcal{CN}(\mathbf{0}_{d_i\times d_i},\frac{P}{d_i}\mathbf{I}_{d_i})$, the following average sum rate is achievable for a given transmit power $P$~\cite{LNIT}:
\begin{align}
R_{\text{sum}}(P)\leq E\left[\sum_{i=1}^{K} \log\frac{\left|\mathbf{A}_{i}(t)+\frac{P}{d_i}\sum_{j=1}^{K}\mathbf{H}^{[e]}_{ij}(t){\mathbf{H}^{[e]}_{ij}(t)}^{\ast}\right|}{\left|\mathbf{A}_{i}(t)+\frac{P}{d_i}\sum_{k=1,i\neq k}^{K}\mathbf{H}^{[e]}_{ik}(t){\mathbf{H}^{[e]}_{ik}(t)}^{\ast}\right|}\right],
\end{align}
where $\mathbf{A}_i(t)=E\left[\mathbf{z}_j^{[e]}(t){\mathbf{z}_j^{[e]}(t)}^\ast\right]=\mathbf{U}^{\ast}_j(t)\mathbf{U}_j'^{\ast}(t)\mathbf{U}_j'(t)\mathbf{U}_j(t)$. Specifically, when all the interferences are eliminated via hybrid beam-forming, i.e., $\mathbf{H}^{[e]}_{ij}(t)=\mathbf{0}$, $\forall i\neq j$ and $\forall t$, (1) becomes
\begin{align}
R_{\text{sum}}(P)&\leq E\left[\sum_{i=1}^{K} \log\frac{\left|\mathbf{A}_{i}(t)+\frac{P}{d_i}\sum_{j=1}^{K}\mathbf{H}^{[e]}_{ij}(t){\mathbf{H}^{[e]}_{ij}(t)}^{\ast}\right|}{\left|\mathbf{A}_{i}(t)\right|}\right]\\
&=\sum_{i=1}^{K}d_i\log(P)+o(\log(P)).
\end{align}

\subsection{Encoding, Decoding, and Sum DoF}

There are $K$ independent messages $W_1,W_2,\ldots, W_K$. For each transmitter $i$, a message $W_i$ is mapped to an $n$ length codeword $\left(\mathbf{x}_{i}(1), \ldots,\mathbf{x}_i(n)\right)$. To send the message $W_i$, at time $t$, transmitter $i$ sends $\mathbf{x}_{i}(t)$. Here, we assume that each transmitter should satisfy the average power constraint $P$, i.e.,
$E[|\mathbf{x}_i(t)|^2]\leq P$ for $i\in\{1,2,\ldots, K\}$. Then receiver $i$ decodes its desired message $\hat{W}_i$, based on its received signal.

A rate tuple
$(R_1, R_2, \ldots, R_K)$ is said to be achievable for the channel if there exists a sequence
of $(2^{nR_1},2^{nR_2}, \ldots ,2^{nR_K}, n)$ codes such that the average probability
of decoding error tends to zero as the code length $n$ goes to
infinity. The capacity region $\mathcal{C}$ of this channel is the
closure of the set of achievable rate tuples $(R_1, R_2, \ldots, R_K)$. The
sum DoF $\Gamma$, which is also known as a sum-capacity pre-log, provides the sum capacity approximation at high SNR as\footnote{In this paper, when we derive lower and upper bounds on the sum DoF, we restrict our attention on the cases in which the hybrid beam-forming structure introduced in Section II-A is used.}
$$C_{\text{sum}}(P)=\max_{(R_1,R_2,\ldots, R_K)\in \mathcal{C}}\sum_{i=1}^{K}R_i(P)=\Gamma \log(P)+o(\log(P)).$$
Equivalently, the sum DoF $\Gamma$ can be defined as $\Gamma=\lim_{P \to
\infty}\max_{(R_1,R_2,\ldots, R_K)\in \mathcal{C}}\frac{\sum_{i=1}^{K}R_i(P)}{\log(P)}$.

\section{Preliminary Discussion}
To gain insights into the DoF gain from hybrid beam-forming, we begin with examining PTP channel, MAC, and BC cases. Note that the PTP channel, the $K$-user MAC, and the $K$-user BC can be obtained from the $K$-user interference channel by allowing full cooperation among all the transmitters and among all the receivers, full cooperation among all the receivers only, and full cooperation among all the transmitters only, respectively. Here, we assume that hybrid beam-forming strategy (including digital precoder and analog precoder) for each channel is employed in a similar manner as in Section II.

\subsection{Point-to-Point (PTP) Channel}
Consider the $(M,M')\times (N,N')$ PTP channel in which the transmitter uses $M$ RF chains and $M'\geq M$ antennas and the receiver uses $N$ RF chains and $N'\geq N$ antennas. The DoF of this channel is stated in the following lemma.
\begin{lemma}\label{Lem:PTP}
For the $(M, M')\times (N, N')$ PTP channel with hybrid beam-forming, the DoF is given by $\Gamma_{\text{PTP}}=\min\{M,N\}$.
\end{lemma}
\begin{proof}
We first provide a converse proof. Following a similar way described in Section II, we can write the input and output relationship of the PTP channel at time slot $t$ as
\begin{align*}
\mathbf{y}(t)&=\mathbf{H}(t) \mathbf{x}(t)+\mathbf{z}(t)\\
&=\mathbf{H}(t)\mathbf{V}'(t)\mathbf{x}^{[b]}(t)+\mathbf{z}(t),
\end{align*}
where $\mathbf{y}(t)$ is the $N'\times 1$ RF-domain output vector at the receiver, $\mathbf{H}(t)$ is the $N'\times M'$ channel matrix from the transmitter to the receiver, $\mathbf{x}(t)$ and $\mathbf{x}^{[b]}(t)$ are the $M'\times 1$ RF-domain input vector and the $M\times 1$ baseband-domain input vector at the transmitter, respectively, $\mathbf{V}'(t)$ is the $M'\times M$ analog precoder of the transmitter, and $\mathbf{z}(t)$ is the $N'\times 1$ Gaussian noise vector.

Now focus on the input and output relationship at baseband domain. By applying receive analog precoding at the receiver, we can get
\begin{align*}
\mathbf{U}'^\ast(t)\mathbf{y}(t)=\mathbf{y}^{[b]}(t)=\mathbf{H}^{[b]}(t) \mathbf{x}^{[b]}(t)+\mathbf{z}^{[b]}(t),
\end{align*}
where $\mathbf{U}'(t)$ is the $N'\times N$ analog precoder of the receiver, $\mathbf{H}^{[b]}(t)=\mathbf{U'}^\ast(t)\mathbf{H}(t)\mathbf{V'}(t)$, and $\mathbf{z}^{[b]}(t)=\mathbf{U}'^\ast(t)\mathbf{z}(t)$. Since $\text{rank}(\mathbf{H}^{[b]}(t))\leq \min\{M,N\}$, we see that $\Gamma_{\text{PTP}}\leq \min\{M,N\}$.

For achievability, we only use $M$ transmit antennas out of $M'$ antennas of the transmitter and $N$ receive antennas out of $N'$ antennas of the receiver to equivalently create a conventional \emph{full digital} PTP channel with $M$ transmit antennas and $N$ receive antennas. Therefore, $\Gamma_{\text{PTP}}\geq \min\{M,N\}$ is achievable~\cite{Foschini98,Telatar99}, which completes the proof.
\end{proof}
\vspace{0.02in}

It is well known that the DoF of the full digital PTP channel with $M$ transmit antennas and $N$ receive antennas is equal to $\min\{M,N\}$~\cite{Foschini98,Telatar99}. Therefore, from the result of Lemma~\ref{Lem:PTP}, we see that adding more antennas only cannot increase the DoF of a PTP channel without increasing the number of RF chains, regardless of the values of $M'$ and $N'$.


\subsection{Multiple Access Channel (MAC) and Broadcast Channel (BC)}
Now we consider the $K$-user MAC and BC with hybrid beam-forming. For the MAC case, each transmitter $i$ uses $M_i$ RF chains and $M_i'\geq M$ antennas and the receiver uses $N$ RF chains and $N'\geq N$ antennas. For the BC case, the transmitter uses  $M$ RF chains and $M'\geq M$ antennas and each receiver uses $N_i$ RF chains and $N_i'\geq N_i$ antennas. The DoFs of these channels are stated in the following lemmas.

\begin{lemma}
For the $K$-user $(M_i, M_i')\times (N, N')$ multiple access channel (MAC) with hybrid beam-forming, the DoF is given by $\Gamma_{\text{MAC}}=\min\left\{\sum_{i=1}^{K} M_i,N\right\}$.
\end{lemma}
\begin{proof}
For a converse proof, we allow full cooperation among all the transmitters to form $\left(\sum_{i=1}^{K}M_i, \sum_{i=1}^{K}M_i'\right)\times (N, N')$ PTP channel. Then, from the result of Lemma 1, the sum DoF of this network is equal to $\min\left\{\sum_{i=1}^{K} M_i,N\right\}$. Since allowing cooperation does not reduce the capacity region~\cite{Viveck1:09}, this is an upper bound of the original network, and thus
\begin{align*}
\Gamma_{\text{MAC}}\leq \min\left\{\sum_{i=1}^{K} M_i,N\right\}.
\end{align*}

For achievability, we use only $M_i$ antennas out of $M'_i$ antennas of transmitter $i$, $\forall i=1,2,\ldots, K$, and $N$ antennas out of $N'$ antennas of the receiver to form a conventional full digital MAC in a similar manner as in Lemma 1. Then, $\Gamma_{\text{MAC}}\geq \min\left\{\sum_{i=1}^{K} M_i,N\right\}$ is achievable~\cite{Jafar07}, which completes the proof.
\end{proof}

\begin{lemma}
For the $K$-user $(M, M')\times (N_i, N_i')$ broadcast channel (BC) with hybrid beam-forming, the DoF is given by $\Gamma_{\text{BC}}=\min\left\{M,\sum_{i=1}^{K} N_i\right\}$.
\end{lemma}
\begin{proof}
We can easily prove Lemma 3 by following similar proof steps in Lemma 2 except the fact that we now allow full cooperation among all the receivers instead of transmitters for a converse proof. For brevity, we omit the rest of the proof steps.
\end{proof}
\vspace{0.02in}

From the results of Lemmas 2 and 3, adding more antennas only without more RF chains cannot increase the sum DoFs of MAC and BC, as in the PTP case.
Therefore, we can see that when full cooperation is already allowed at either transmitter side or receiver side of the $K$-user interference channel, hybrid beam-forming cannot further improve the DoF. However, as we will show in the following example, for the case in which full cooperation is not allowed so that there exist inter-user interferences, the sum DoF of an interference channel can be improved via hybrid beam-forming for certain cases.

\subsection{Interference Channel: Motivating Example}
Now we provide a simple example where hybrid beam-forming indeed improves the sum DoF. In the following example, we omit the time index $t$ for brevity.

\begin{example}
Consider the two-user $(2,4)\times (2,2)$ interference channel where $M_i=N_i=N_i'=2$ and $M_i'=4$, $\forall i=1,2$. We first set the $4\times 2$ analog precoder $\mathbf{V}'_i$ to satisfy $\mathbf{H}_{ji}\mathbf{V}'_i=\mathbf{0}$ for $i\neq j$ and $\text{rank}(\mathbf{H}_{ii}\mathbf{V}'_i)=2$. Since $\mathbf{H}_{ji}$ is the $2\times 4$ matrix and all channel coefficients are generic, we can easily find $\mathbf{V}_i'$ that satisfies these conditions. In addition, for the digital precoder of transmitter $i$, we set $\mathbf{V}_i=\mathbf{I}_2$, $\forall i=1,2$. Then, the received signal at each receiver $i$ is given by
\begin{align*}
\mathbf{y}_i&=\mathbf{H}_{ii}\mathbf{V}'_i\mathbf{V}_i\mathbf{s}_i+\mathbf{H}_{ij}\mathbf{V}'_j\mathbf{V}_j\mathbf{s}_j+\mathbf{z}_i\\
&=\mathbf{H}_{ii}\mathbf{V'}_i\mathbf{s}_i+\mathbf{z}_i,
\end{align*}
where $\mathbf{s}_i\sim \mathcal{CN}(\mathbf{0}_{2\times 2},\frac{P}{2}\mathbf{I}_2)$ is the transmitted symbol vector of user $i$ and $i\neq j$. Since $\text{rank}(\mathbf{H}_{ii}\mathbf{V}'_i)=2$, we can achieve $d_i=2$ for each user, thus achieving $\Gamma\geq 4$. Note that for the two-user full digital $(2,2)\times (2,2)$ interference channel, which has the same number of RF chains as in the two-user $(2,4)\times (2,2)$ interference channel, only the sum DoF of two can be achieved. This shows that for some cases, the sum DoF of an interference channel can actually be increased by adding more antennas only without increasing the number of RF chains.
\end{example}
\vspace{0.02in}
\begin{remark}
As shown in Example 1, by using more antennas, we can have a better ability to null out interferences from/to other users at RF domain. This enables users to secure more interference-free dimensions, and as a result, a higher sum DoF is achievable without any additional RF chains for some cases. However, despite this improved capability dealing with interferences, hybrid beam-forming does not always increase the DoF of an interference channel. For instance, as will be demonstrated in the next example, if all the interferences can be eliminated without the need to add more antennas, hybrid beam-forming cannot increase the sum DoF.
\end{remark}
\begin{example}
Consider the two-user $(1,2)\times (2,4)$ interference channel where $M_i=1$, $M'_i=2$, $N_i=2$, and $N'_i=4$, $\forall i=1,2$. By allowing full cooperation among transmitters and among receivers, we can get the $(2,4)\times (4,8)$ PTP channel. Since the DoF of this channel is given by two from Lemma 1 and allowing full cooperation does not reduce the capacity region, the sum DoF of the two-user $(1,2)\times (2,4)$ interference channel cannot be more than two. Note that the two-user full digital $(1,1)\times (2,2)$ interference channel can also achieve the sum DoF of two~\cite{Jafar07}. Therefore, unlike in Example 1, adding antennas only cannot increase the sum DoF in this case. In fact, in this case, to achieve a higher DoF, we need to use more RF chains as well as more antennas. For example, if we use additional one RF chain and two RF chains at each transmitter and receiver, respectively, the channel becomes the two-user full digital $(2,2)\times (4,4)$ interference channel, and we can now achieve the improved DoF of 4.
\end{example}

\section{Main Results}
In this section, we state and discuss about the main results of this paper. For the two-user case, the sum DoF is completely characterized for any antenna configuration. When $K\geq 3$, we focus on a symmetric case where $M_i=M$, $N_i=N$, $M_i'=M'$, and $N'_i=N'$, $\forall i=1,2,\ldots,K$, and derive lower and upper bounds on the sum DoF. It is shown that two bounds are matched under a certain condition.

\subsection{Two-user Case}

For the two-user interference channel, we completely characterize the sum DoF as stated in the following theorem.

\begin{theorem}[Two-user case]
For the two-user $(M_i, M_i')\times (N_i, N_i')$ interference channel with hybrid beam-forming, the sum DoF $\Gamma$ is given by
\begin{align*}
\Gamma=\min\{M_1+M_2, N_1+N_2, M_1+N_2, M_2+N_1, \max\{M_1',N_2'\}, \max\{M_2',N_1'\}\}
\end{align*}
where $M_i\leq M_i'$ and $N_i\leq N_i'$ for $i=1,2$.
\end{theorem}
\begin{proof}
See Section VI for the proof.
\end{proof}
\vspace{0.02in}

\begin{remark}
For the case where $M_i'=M_i$ and $N_i'=N_i$, $\forall i=1,2$, the sum DoF becomes
\begin{align*}
\Gamma=\min\{M_1+M_2, N_1+N_2,\max\{M_1,N_2\}, \max\{M_2,N_1\}\},
\end{align*}
which recovers the result for the two-user full digital interference channel in~\cite{Jafar07}.
\end{remark}

\begin{figure}[t!]
\begin{center}
\includegraphics[scale=0.8]{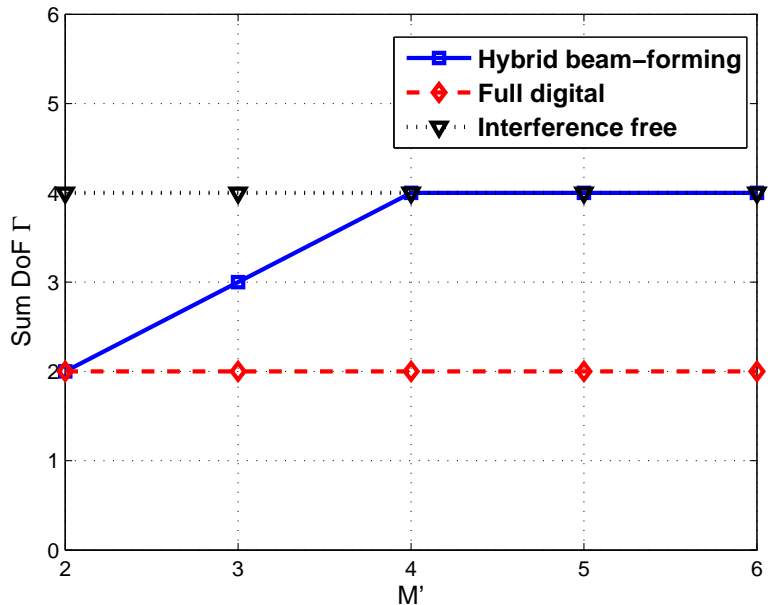}
\end{center}
\vspace{-0.15in}
\caption{Sum DoFs of the two-user case with respect to $M'$ when $M_1=M_2=N_1=N_2=M=2$ and $M_1'=M_2'=N_1'=N_2'=M'$.}
\label{example1}
\vspace{-0.1in}
\end{figure}

\begin{remark}
Note that when the condition
$$\min\{\max\{M'_1,N'_2\}, \max\{M'_2,N'_1\}\}\geq \min\{M_1+M_2, N_1+N_2, M_1+N_2, M_2+N_1\}$$
is satisfied, the sum DoF becomes
\begin{align*}
\Gamma&=\min\{M_1+M_2, N_1+N_2, N_1+M_2, M_2+N_1\}\\
&=\min\{M_1, N_1\}+\min\{M_2,N_2\},
\end{align*}
which is the sum DoF of the interference-free channel. Therefore, we can see that by adding enough number of antennas at each node, all the users can utilize their full DoFs as if there is no interference.
\end{remark}
\vspace{0.05in}

\textbf{DoF gain due to hybrid beam-forming:} Consider a symmetric case where $M_1=M_2=N_1=N_2=M=2$ and $M_1'=M_2'=N_1'=N_2'=M'$. We plot the sum DoF as a function of $M'$ with fixed $M$ in Fig.~\ref{example1}. For comparison, we also plot the sum DoF of the full digital case where the number of RF chains is the same as the hybrid beam-forming case. As can be seen in Fig.~\ref{example1}, although we add antennas only, we can achieve a higher DoF and it reaches up to the maximum value of $2M$, the sum DoF of the interference-free channel, when $M'=2M$. The gain comes from the fact that hybrid beam-forming can null out more interferences without increasing the number of RF chains, as well as enhancing the capacity of PTP channel as reported in~\cite{ElAyachRajagopalAbu-SurraPiHeath:TWC14,WuChiuLinChang:TWC13}.

\subsection{$K$-user Case}
When $K\geq 3$, we focus on a symmetric case where $M_i=M$, $N_i=N$, $M_i'=M'$, and $N'_i=N'$, $\forall i=1,2,\ldots,K$.
Under this configuration, we obtain lower and upper bounds on the sum DoF as stated in the following theorem, which are tight when $\frac{\max\{M',N'\}}{\min\{M',N'\}}$ is an integer.

\begin{theorem}[$K$-user case]
For the symmetric $K$-user $(M, M')\times (N, N')$ interference channel with hybrid beam-forming, the following sum DoF $\Gamma$ is achievable:
\begin{align*}
\Gamma\geq \left\{\begin{array}{ll} K\min\{M,N\}& \textrm{if $K\leq R $}, \\
K\min \left\{M,N,\frac{R}{R+1}\min\{M',N'\}\right\}& \textrm{if $K> R$,}
\end{array} \right.
\end{align*}
where $R=\left\lfloor\frac{\max\{M',N'\}}{\min\{M',N'\}}\right\rfloor$. For converse, the sum DoF $\Gamma$ is upper bounded by
\begin{align*}
\Gamma\leq \left\{\begin{array}{ll} K\min\{M,N\}& \textrm{if $K\leq R $}, \\
K\min\left\{M,N,\frac{\max\{M',N'\}}{R+1}\right\}& \textrm{if $K> R$.}
\end{array} \right.
\end{align*}
\begin{proof}
See Section VII for the proof.
\end{proof}
\end{theorem}
\vspace{0.02in}

\begin{remark}
Similar to the two-user case explained in Remark 2, for the case where $M'=M$ and $N'=N$, Theorem 2 recovers the result for the $K$-user full digital interference channel in~\cite{Tiangao:10}.
\end{remark}

\begin{remark}
It is easy to see that $\frac{R}{R+1}\min\{M',N'\}$ is a non-decreasing function of $M'$ and $N'$. Intuitively, this is clear since having more antennas does not reduce the capacity region. Moreover, when $\frac{R}{R+1}\min\{M',N'\}\geq \min\{M,N\},$ each user can achieve the maximum DoF of $\min\{M,N\}$ as if there is no interference.
\end{remark}

\begin{cor}
By employing hybrid beam-forming, we can get at most \emph{two-fold} DoF gain as compared to the full digital case in which the number of RF chains is the same as the hybrid beam-forming case.
\begin{proof}
Let $\Gamma_h$ and $\Gamma_f$ denote the sum DoFs with hybrid beam-forming and full digital structures, respectively. For the two-user case, we have
\begin{align*}
\frac{\Gamma_h}{\Gamma_f}&=\frac{\min\{M_1+M_2, N_1+N_2, M_1+N_2, M_2+N_1, \max\{M_1',N_2'\}, \max\{M_2',N_1'\}\}}{\min\{M_1+M_2, N_1+N_2,\max\{M_1,N_2\}, \max\{M_2,N_1\}\}}\\
&\leq \frac{\min\{M_1,N_1\}+\min\{M_2,N_2\}}{\min\{M_1+M_2, N_1+N_2,\max\{M_1,N_2\}, \max\{M_2,N_1\}\}}\\
&\leq \frac{\min\{M_1,N_1\}+\min\{M_2,N_2\}}{\max\{\min\{M_1, N_1\},\min\{M_2,N_2\}\}}\\
&\leq \frac{2\max\{\min\{M_1,N_1\},\min\{M_2,N_2\}\}}{\max\{\min\{M_1, N_1\},\min\{M_2,N_2\}\}}\\
&=2
\end{align*}
In addition, for the general $K$-user case, we have
\begin{align*}
\frac{\Gamma_h}{\Gamma_f}&\leq \frac{K\min\{M,N\}}{\frac{KL}{L+1}\min\{M,N\}}\\
&=\frac{L+1}{L}\\
&\leq 2,
\end{align*}
where $L=\left\lfloor\frac{\max\{M,N\}}{\min\{M,N\}}\right\rfloor$. This completes the proof.
\end{proof}
\end{cor}
\vspace{0.05in}

\begin{figure}[t!]
\begin{center}
\includegraphics[scale=0.8]{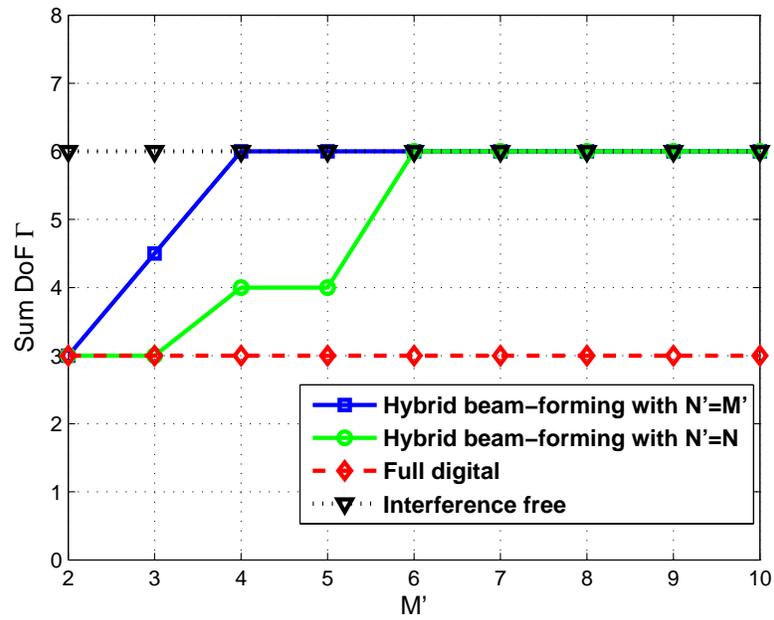}
\end{center}
\vspace{-0.15in}
\caption{Sum DoFs of the three-user cases with respect to $M'$ when $M=N=2$.}
\label{example2}
\vspace{-0.1in}
\end{figure}

\vspace{0.02in}
\textbf{DoF gain due to hybrid beam-forming:} Consider the three-user case where $M=N=2$. First, we set $N'=M'$ and plot the sum DoF as a function of $M'$ with fixed $M$ and $N$ in Fig.~\ref{example2}. In addition, we consider another scenario in which additional antennas are employed only at transmitters, i.e., $N'=N$, and again plot the sum DoF as a function of $M'$. As can be seen in the figure, by using hybrid beam-forming, we can achieve a higher DoF and interestingly, it can reach up to the maximum DoF of six even when hybrid beam-forming is applied at transmitters only. Furthermore, note that when achieving this DoF, interference alignment combined with hybrid beam-forming is employed. From this point, we can see that hybrid beam-forming can provide an improved capability not only nulling out interferences but also aligning interferences at RF domain.

\begin{figure}[t!]
\begin{center}
\includegraphics[scale=0.8]{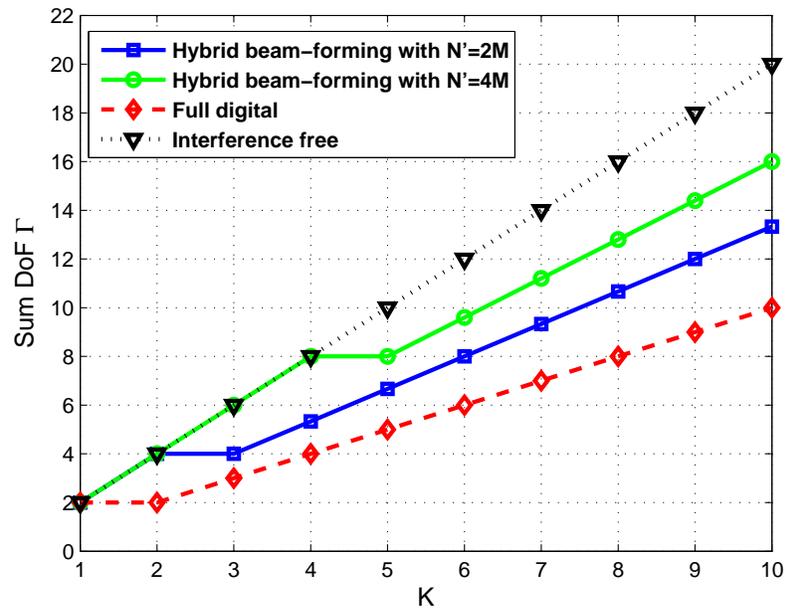}
\end{center}
\vspace{-0.15in}
\caption{Sum DoFs with respect to $K$ when $M=N=M'=2$.}
\label{example3}
\vspace{-0.1in}
\end{figure}

Now, we examine a tendency of the sum DoF with respect to $K$ with the fixed number of antennas and RF chains at each node. Specifically, we set $M=N=M'=2$ and plot the sum DoFs when $N'=2M$ and $N'=4M$ in Fig.~\ref{example3}. For comparison, we also plot the sum DoF of the full digital case where the number of RF chains is the same as the hybrid beam-forming case. From Fig.~\ref{example3}, we see that hybrid beam-forming can improve the sum DoF for all values of $K$, and moreover, the slope also increases as the number of additional antennas at each receiver increases.

\section{Numerical Simulation}
In this section, we numerically evaluate the average sum rate performance of the proposed hybrid beam-forming schemes for $K=2$ and $3$ cases to show that the sum DoFs stated in Theorems 1 and 2 are indeed achievable. For comparison, the sum DoFs of the full digital and the interference-free cases are also plotted. Here, we assume Rayleigh fading environment where each channel coefficient is drawn i.i.d from $\mathcal{CN}(0,1)$. In addition, we assume that all the noise power is normalized to unity and thus $\textrm{SNR}=P$. Furthermore, to clearly capture the sum DoFs from the sum-rate graphs, we plot the average sum rates as a function of $\log_2(\textrm{SNR})$.

\begin{figure}[t!]
\begin{center}
\includegraphics[width=0.7\textwidth]{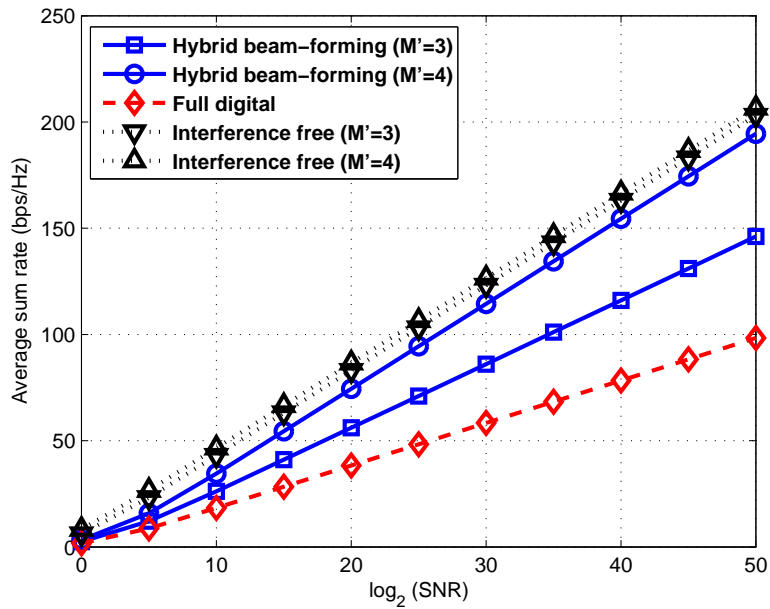}
\end{center}
\caption{Average sum rates of the two-user case when $M=N=2$ and $N'=M'$.}
\label{Fig:SumRate-TwoUser}
\end{figure}

\subsection{Average Sum Rate for the Two-user Case}
In Fig.~\ref{Fig:SumRate-TwoUser}, the average sum rates are plotted as a function of $\log_2(\textrm{SNR})$, where $M=N=2$ and $N'=M'$. Note that the sum DoFs can be observed from the slopes in the figure. We can see that the sum DoFs obtained by the simulation are well matched with the sum DoFs stated in Lemma 1 and Theorem 1.
Here, when the simulation is performed, the number of streams of hybrid beam-forming for each user is set by $d_1=2$ and $d_2=1$ for $N'=M'=3$ and $d_1=d_2=2$ for $N'=M'=4$ by following Theorem 1.

As shown in the figure, the full digital scheme can only achieve the sum DoF of two, while the sum DoF of the interference-free channel is four. When hybrid beam-forming is employed, we can see by simulation that the sum DoF can be improved and even reach up to the interference-free DoF, as shown in Theorem 1, and therefore the performance gap between hybrid beam-forming and full digital cases dramatically increases as the SNR increases.
%

\begin{figure}[t!]
\begin{center}
\includegraphics[width=0.7\textwidth]{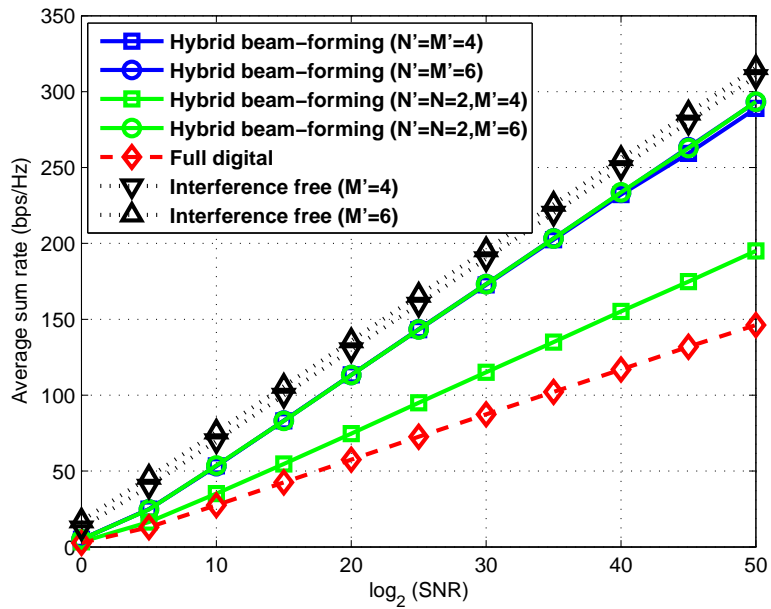}
\end{center}
\caption{Average sum rates of the three-user case when $M=N=2$.}
\label{Fig:SumRate-ThreeUser}
\end{figure}

\subsection{Average Sum Rate for the Three-user Case}
As in the previous subsection, the average sum rate is plotted as a function of $\log_2(\textrm{SNR})$ in Fig.~\ref{Fig:SumRate-ThreeUser}, where $M=N=2$. When hybrid beam-forming is used, we consider the two different scenarios in which additional antennas are employed only at transmitters, i.e., $N'=N=2$ for $M'=4$ and 6, and additional antennas are employed both at transmitters and receivers, i.e., $N'=M'$ for $M'=4$ and 6. Here, we adopt the distributed interference alignment\footnote{Note that the achievable scheme proposed in Theorem 2 requires an arbitrary large number of symbol extension. Therefore, in this subsection, instead of adopting the achievable scheme in Theorem 2 directly, we employ the DIA algorithm to numerically show that the sum DoF stated in Theorem 2 is indeed feasible. Here, Theorem 2 provides theoretical guidance when selecting a suitable number of streams for each user.}   (DIA) algorithm proposed in~\cite{DIA} for numerical simulation and the number of streams of hybrid beam-forming used for the simulation is given by Theorem 2. The slopes in the figure show that the sum DoFs stated in Theorem 2 is indeed achievable.

The full digital scheme can only achieve the sum DoF of three, while the sum DoF of the interference-free channel is six as shown in the figure. As in the two-user case, the sum DoF of the full digital scheme is only half of that of the interference-free channel. When $N'=M'=4$ and $6$, the hybrid beam-forming can achieve the maximum sum DoF of six as if there is no interference between users. Interestingly, for the case in which additional antennas are employed only at transmitters $(N'=N=2,M'=4,6)$, the sum DoF can also be increased as compared to the full digital case, and the performance gain over the full digital case increases as the number of additional antennas increases.




\section{Proof of Theorem 1}
\subsection{Achievability}
In our achievable scheme, we will use only $d_i$ transmit RF chains out of $M_i$ RF chains of transmitter $i$ and $d_i$ receive RF chains out of $N_i$ RF chains of receiver $i$, for all $i=1,2$. Hence, from now on, we can equivalently consider the $(d_i, M_i')\times (d_i, N_i')$ interference channel instead of the original channel, the $(M_i, M_i')\times (N_i, N_i')$ interference channel. In addition, since our achievable scheme operates in a single time slot, we omit the time index $t$ for brevity.

We design the input signal of transmitter $i$ as
\begin{align*}
\mathbf{x}_i=\mathbf{V}'_i\mathbf{V}_i\mathbf{s}_i,
\end{align*}
where $\mathbf{V}'_i$ is the $M_i'\times d_i$ transmit analog precoder, $\mathbf{V}_i$ is the $d_i\times d_i$ transmit digital precoder, and $\mathbf{s}_i\sim \mathcal{CN}\left(\mathbf{0}_{d_i\times d_i},\frac{P}{d_i}\mathbf{I}_{d_i}\right)$ is the $d_i\times 1$ vector of transmitted Gaussian symbols of user $i$. To be specific, beam-forming vectors in $\mathbf{V}'_i$ can be decomposed into two parts:
\begin{align*}
\mathbf{V}'_{i}=\left[\begin{array}{ccccccccccccccc}
\mathbf{V}'_{ii}& \mathbf{V}'_{i0}
\end{array}\right]
\end{align*}
\begin{itemize}
\item $\mathbf{v}'_{ii,k}$ denotes the $k$th beam-forming vector in $\mathbf{V}'_{ii}$ such that $\mathbf{H}_{ii}\mathbf{v}'_{ii,k}\neq 0$ and $\mathbf{H}_{ji}\mathbf{v}'_{ii,k}= 0$, where $i\neq j$. Note that since the size of $\mathbf{H}_{ji}$ is given by $N'_j\times M'_i$ and channel matrices are drawn i.i.d from a continuous distribution, the maximum number of linearly independent beam-forming vectors satisfying this condition is $\max(0,M_i'-N_j')$. Let $d_{ii}\leq \max(0,M_i'-N_j')$ denote the number of such vectors.
\item $\mathbf{v}'_{i0,l}$ denotes the $l$th beam-forming vector in $\mathbf{V}'_{i0}$ whose coefficients are randomly generated from a continuous distribution and $0<||\mathbf{v}'_{i0,l}||\leq \alpha$, where $\alpha$ has a finite value. Hence, $\mathbf{H}_{ii}\mathbf{v}'_{i0,l}\neq0$ and $\mathbf{H}_{ji}\mathbf{v}'_{i0,l}\neq0$ for $i\neq j$ with probability one. Let $d_{i0}=d_{i}-d_{ii}$ denote the number of such vectors. In addition, we further restrict $d_i$ and $d_{j0}$ to satisfy $d_{i}+d_{j0}\leq N_i'$.
\end{itemize}
In summary, we choose $d_1$, $d_{11}$, $d_{10}$, $d_2$, $d_{22}$, and $d_{20}$ to satisfy the following conditions.
\begin{align}
&0\leq d_1=d_{11}+d_{10}\leq \min(M_1,N_1)\\
&0\leq d_2=d_{22}+d_{20}\leq \min(M_2,N_2)\\
&0\leq d_{11}\leq \max(0,M_1'-N_2')\\
&0\leq d_{22}\leq \max(0,M_2'-N_1')\\
&0\leq d_1+d_{20}\leq N_1'\\
&0\leq d_2+d_{10}\leq N_2'
\end{align}

Then the received signal at receiver $i\in\{1,2\}$ at RF domain is given by
\begin{align}
\mathbf{y}_i&=\mathbf{H}_{ii}\mathbf{x}_i+\mathbf{H}_{ij}\mathbf{x}_j+\mathbf{z}_i \nonumber \\
&=\mathbf{H}_{ii}\mathbf{V}'_i\mathbf{V}_i\mathbf{s}_i+\mathbf{H}_{ij}\mathbf{V}'_j\mathbf{V}_j\mathbf{s}_j+\mathbf{z}_i \nonumber \\
&=\mathbf{H}_{ii}\mathbf{V}'_i\mathbf{V}_i\mathbf{s}_i+\mathbf{H}_{ij}\left[\begin{array}{cc}\mathbf{0}_{M'_j\times d_{jj}}&\mathbf{V}'_{j0}\end{array}\right]\mathbf{V}_j\mathbf{s}_j+\mathbf{z}_i,
\end{align}
where (10) is due to the properties of $\mathbf{V}'_{jj}$ and $\mathbf{V}'_{j0}$.

Now we explain the beam-forming matrix at receiver $i$. Denote $\mathbf{U'}_i$ as the $N_i'\times d_i$ receive analog precoder and $\mathbf{U}_i$ as the $d_i\times d_i$ receive digital precoder. We set $\mathbf{U'}_i$ such that $\mathbf{U'}_i^{\ast}\mathbf{H}_{ij}\left[\begin{array}{cc}\mathbf{0}_{M'_j\times d_{jj}}&\mathbf{V}'_{j0}\end{array}\right]\mathbf{V}_j=\mathbf{0}$ and $\text{rank}(\mathbf{U'}_i^{\ast}\mathbf{H}_{ii}\mathbf{V}'_i\mathbf{V}_i)=d_i$. Since we have
\begin{align*}
&\text{rank}\left(\mathbf{H}_{ii}\mathbf{V}'_{i}\mathbf{V}_i\right)=d_{i}\\
&\text{rank}\left(\mathbf{H}_{ij}\left[\begin{array}{cc}\mathbf{0}_{M'_j\times d_{jj}}&\mathbf{V}'_{j0}\end{array}\right]\mathbf{V}_j\right)=d_{j0}\\
&d_i+d_{j0}\leq N_i',
\end{align*}
we can find $\mathbf{U'}_i$ satisfying these conditions. Therefore, after applying receive analog precoding, we obtain
\begin{align*}
\mathbf{U'}_i^{\ast}\mathbf{y}_i&=\mathbf{U'}^{\ast}_i\mathbf{H}_{ii}\mathbf{V}'_i\mathbf{V}_i\mathbf{s}_i+\mathbf{U'}^{\ast}_i\mathbf{H}_{ij}\left[\begin{array}{cc}\mathbf{0}_{M'_j\times d_{jj}}&\mathbf{V}'_{j0}\end{array}\right]\mathbf{V}_j\mathbf{s}_j+\mathbf{U'}^{\ast}_i\mathbf{z}_i\\
&=\mathbf{U'}^{\ast}_i\mathbf{H}_{ii}\mathbf{V}'_i\mathbf{V}_i\mathbf{s}_i+\mathbf{U'}^{\ast}_i\mathbf{z}_i.
 \end{align*}
Recall that $\text{rank}(\mathbf{U'}^{\ast}_i\mathbf{H}_{ii}\mathbf{V}'_i\mathbf{V}_i)=d_i$. Now, we set $\mathbf{U}_i$ and $\mathbf{V}_i$ as the left and right singular matrices of the matrix $\mathbf{U'}^{\ast}_i\mathbf{H}_{ii}\mathbf{V}'_i$, respectively. Then we get $d_i$ parallel AWGN channels for user $i$ after applying the receive digital precoding as follows:
\begin{align*}
\mathbf{U}^{\ast}_i\mathbf{U'}_i^{\ast}\mathbf{y}_i=\mathbf{y}_i^{[e]}&=\mathbf{U}^{\ast}_i\mathbf{U'}^{\ast}_i\mathbf{H}_{ii}\mathbf{V}'_i\mathbf{V}_i\mathbf{s}_i+\mathbf{U}^{\ast}_i\mathbf{U'}^{\ast}_i\mathbf{z}_i\\
&=\Lambda_i \mathbf{s}_i +\mathbf{z}^{[e]}_i,
\end{align*}
where $\Lambda_i$ is the $d_i\times d_i$ diagonal matrix with the singular values of $\mathbf{U'}^{\ast}_i\mathbf{H}_{ii}\mathbf{V}'_i$ on the diagonal and $\mathbf{z}^{[e]}_i=\mathbf{U}^{\ast}_i\mathbf{U'}^{\ast}_i\mathbf{z}_i$. Therefore, we can see that each user achieves $d_i$ DoF via the proposed scheme, and thus the achievable total DoF is given by $\Gamma\geq d_1+d_2$.

Finally, by evaluating the conditions (4)--(9) using the Fourier-Motzkin elimination, we get the desired bound:
\begin{align*}
\Gamma\geq\min\{M_1+M_2, N_1+N_2, M_1+N_2, M_2+N_1, \max\{M_1',N_2'\}, \max\{M_2',N_1'\}\},
\end{align*}
which completes the achievability proof of Theorem 1.

\subsection{Converse}
From the result of Lemma 1, the DoF of the $(M_i,M_i')\times (N_i,N_i')$ PTP channel for each user $i$ is equal to $\min\{M_i,N_i\}$. Therefore, for the two-user $(M_i,M_i')\times (N_i,N_i')$ interference channel, the sum DoF cannot be more than $\sum_{i=1}^{2}\min\{M_i,N_i\}$, i.e.,
\begin{align}
\Gamma \leq \min\{M_1+M_2,M_1+N_2, N_1+M_2, N_1+N_2\}.
\end{align}

Now suppose we add $M_i'-M_i$ transmit RF chains at transmitter $i$ and $N'_i-N_i$ receive RF chains at receiver $i$ for all $i=1,2$ to form the conventional full digital $(M_i',M_i')\times(N_i',N_i')$ interference channel. Then the sum DoF $\Gamma_f$ of this channel is upper bounded by
\begin{align}
\Gamma_f\leq \min\{M'_1+M'_2, N'_1+N'_2, \max\{M'_1,N'_2\}, \max\{M'_2,N'_1\}\}
\end{align} from the result of~\cite{Jafar07}.
Clearly, adding more RF chains does not reduce the capacity region, and hence (12) is also an upper bound for the original channel.

Combining (11) and (12), we get the desired upper bound as
\begin{align*}
\Gamma\leq \min\{M_1+M_2,M_1+N_2, N_1+M_2, N_1+N_2, \max\{M'_1,N'_2\}, \max\{M'_2,N'_1\}\},
\end{align*}
which completes the converse proof of Theorem 1.

\section{Proof of Theorem 2}
\subsection{Achievability}
Our achievability is motivated by the interference alignment scheme proposed for the $K$-user full digital $(M,M) \times (N,N)$ interference channel in~\cite{Tiangao:10}. Here, we extend the previous scheme to be suitable for the general $K$-user $(M,M') \times (N,N')$ interference channel with hybrid beam-forming. For brevity, we focus on explaining the steps needed for hybrid beam-forming cases.

Consider the ratio $R=\left\lfloor\frac{\max\{M',N'\}}{\min\{M',N'\}}\right\rfloor$. Similar in~\cite{Tiangao:10}, when $K\leq R$, our achievable scheme is based on zero forcing while it is based on interference alignment when $K>R$. Note that reciprocity holds for both zero forcing and interference alignment, i.e., the achievable sum DoF of the $K$-user $(M,M') \times (N,N')$ interference channel via zero forcing and/or interference alignment is equal to the that of the $K$-user $(N,N') \times (M,M')$ interference channel~\cite{Cadambe107,DIA}. Therefore, without loss of generality, we assume that $M'\leq N'$, which results in $R=\left\lfloor\frac{N'}{M'}\right\rfloor$.


\subsubsection{$K\leq R$}
In this case, since our achievable scheme operates in a single time slot, we omit the time index $t$ for brevity.

Each transmitter sends $d=\min\{M,N\}$ data streams using hybrid beam-forming, i.e.,
\begin{align*}
\mathbf{x}_i=\mathbf{V}'_i\mathbf{V}_i\mathbf{s}_i,
\end{align*}
where $\mathbf{V}'_i$ is the $M_i'\times d$ transmit analog precoder, $\mathbf{V}_i$ is the $d\times d$ transmit digital precoder, and $\mathbf{s}_i\sim \mathcal{CN}\left(\mathbf{0}_{d\times d},\frac{P}{d}\mathbf{I}_{d}\right)$ is the $d\times 1$ vector of transmitted Gaussian symbols of user $i$. Here we set that coefficients of $\mathbf{V}'_i$ and $\mathbf{V}_i$ are randomly generated from a continuous distribution and $0<||\mathbf{v}'_{i0,l}||\leq \alpha$, where $\alpha$ has a finite value. Then the received signal at receiver $i\in\{1,2,\ldots,K\}$ at RF domain is given by
\begin{align*}
\mathbf{y}_i&=\mathbf{H}_{ii}\mathbf{x}_i+\sum_{j=1,j\neq i}^{K}\mathbf{H}_{ij}\mathbf{x}_j+\mathbf{z}_i\\
&= \mathbf{H}_{ii}\mathbf{V}'_i\mathbf{V}_i\mathbf{s}_i+\sum_{j=1,j\neq i}^{K}\mathbf{H}_{ij}\mathbf{V}'_j\mathbf{V}_j\mathbf{s}_j+\mathbf{z}_i.
\end{align*}
Observe that $\text{rank}(\mathbf{H}_{ii}\mathbf{V}'_i\mathbf{V}_i)=d$ and $\text{rank}([\begin{array}{ccc}\mathbf{H}_{i1}\mathbf{V}'_1\mathbf{V}_1&\cdots &\mathbf{H}_{iK}\mathbf{V}'_K\mathbf{V}_K\end{array}])=Kd$, $\forall i=1,2,\ldots,K$. Since $Kd=K\min\{M,N\}\leq KM'\leq RM'\leq N'$, we can completely null out all the interference at each receiver $i$ by setting analog beam-forming matrix, $\mathbf{U'}_i$, as
\begin{align*}
\mathbf{U'}_i^{\ast}\left[\begin{array}{ccccccc}\mathbf{H}_{i1}\mathbf{V}'_1\mathbf{V}_1&\cdots&\mathbf{H}_{i,i-1}\mathbf{V}'_{i-1}\mathbf{V}_{i-1} &\mathbf{H}_{i,i+1}\mathbf{V}'_{i+1}\mathbf{V}_{i+1}&\cdots& \mathbf{H}_{iK}\mathbf{V}'_{K}\mathbf{V}_{K}\end{array}\right]=\mathbf{0},
\end{align*}
while $\text{rank}(\mathbf{U}_i^{\ast}\mathbf{H}_{ii}\mathbf{V}'_i\mathbf{V}_i)=d$ can also be satisfied for the desired signals. As a result, each user can achieve $d$ DoF, and thus achieving $\Gamma\geq Kd=K\min\{M,N\}$.

\subsubsection{$K>R$}
In this case, before we explain our achievable scheme, we first refer to the following Lemma in~\cite{Tiangao:10}.
\begin{lemma}
For the $K(>R+1)$-user full digital $(1,1)\times (R,R)$ single--input multiple--output (SIMO) interference channel, the sum DoF of $\frac{R}{R+1}K$ can be achieved.
\end{lemma}
\begin{proof}
The proof is provided in~\cite[Theorem 2]{Tiangao:10}.
\end{proof}

Now consider the $KM'(>R+1)$--user full digital $(1,1)\times (R,R)$ SIMO interference channel. By adapting the achievable scheme in Lemma 4, we can achieve the sum DoF of $\frac{R}{R+1}KM'$. To be specific, under this scheme, $T=(R+1)(n+1)^p$ symbol extension of the original channel is considered, where $p=M'KR(M'K-R-1)$ and $n\in\mathbb{N}$ is an arbitrary integer, and each user $i\in \{1,2,\ldots, R+1\}$ achieves $d^{s}_i=R(n+1)^p$ DoF and each user $i\in \{R+2,R+3,\ldots, M'K\}$ achieves $d^{s}_i=Rn^p$ DoF over the extended channel, i.e.,
\begin{align*}
d_i^s= \left\{\begin{array}{ll} R(n+1)^p & \textrm{if $i\leq R+1 $}, \\
Rn^p & \textrm{if $i> R+1$.}
\end{array} \right.
\end{align*}
In addition, by applying the scheme, it turns out that the dimensions of the signal space spanned by the desired signal vectors and interference signal vectors at receiver $i\in \{1,2,\ldots, R+1\}$ are given by $R(n+1)^p$ and $R^2(n+1)^p$ out of the $RT=R(R+1)(n+1)^p$ dimensional signal space, respectively, while they are given by $Rn^p$ and $R(R+1)(n+1)^p-Rn^p$, respectively, at receiver $i\in \{R+2,R+3,\ldots, KM'\}$. Let $\mathbf{\tilde{V}}'_i$ denote the $T\times d^{s}_i$ beam-forming matrix of transmitter $i$ used for the extended channel of the $KM'$-user $(1,1)\times (R,R)$ interference channel. We denote the elements of $\mathbf{\tilde{V}}'_i$ as
\begin{align*}
\mathbf{\tilde{V}}'_i=\left[\begin{array}{ccccccc}v_{i,1}(1) & v_{i,2}(1) & \cdots & v_{i,d_i^s}(1) \\
v_{i,1}(2) & v_{i,2}(2) & \cdots & v_{i,d_i^s}(2)\\
\vdots & \vdots & \ddots & \vdots\\
v_{i,1}(T) & v_{i,2}(T) & \cdots & v_{i,d_i^s}(T)\\
  \end{array}\right],
\end{align*}
where $v_{i,j}(t)$ means the $j$th beam-forming coefficient of transmitter $i$ at time slot $t$.

Then, now consider the original channel, the $K$-user $(M, M')\times (N,N')$ interference channel. Here, we only use $RM'$ antennas out of $N'$ antennas at each receiver by discarding $N'-RM'$ antennas at each receiver, which results in the $K$-user $(M,M')\times (N,RM')$ interference channel, and then apply the $T$-time symbol extension as in the $KM'$-user $(1,1)\times (R,R)$ interference channel, which gives the overall channel matrix between transmitter $i$ and receiver $j$ as
\begin{align*}
\mathbf{\bar{H}}_{ij}=\left[\begin{array}{ccccccc}\mathbf{H}_{ij}(1) &\mathbf{0}_{RM'\times M'} & \cdots & \mathbf{0}_{RM'\times M'}\\
\mathbf{0}_{RM'\times M'}& \mathbf{H}_{ij}(2)  & \cdots & \mathbf{0}_{RM'\times M'}\\
\vdots & \vdots & \ddots & \vdots\\
\mathbf{0}_{RM'\times M'} & \mathbf{0}_{RM'\times M'} & \cdots & \mathbf{H}_{ij}(T) \\
  \end{array}\right].
\end{align*}
For this extended channel, by employing beam-forming coefficients proposed in the $KM'$-user $(1,1)\times (R,R)$ interference channel, we design the analog beam-forming matrix of transmitter $i\in\{1,2,\ldots,K\}$ as
\begin{align*}
\mathbf{{\bar{V}}}'_i=\left[\begin{array}{ccccccc}\mathbf{\bar{\bar{V}}'}_i(1)\\
\mathbf{\bar{\bar{V}}'}_i(2)  \\
\vdots \\
\mathbf{\bar{\bar{V}}'}_i(T)\\
\end{array}\right]
\end{align*}
where
\begin{align*}
\mathbf{\bar{\bar{V}}'}_i(t)=\left[\begin{array}{ccccccc}\mathbf{v}_{M'(i-1)+1}(t) &\mathbf{0}_{1\times d_{M'(i-1)+2}^s}& \cdots & \mathbf{0}_{1\times d_{M'i}^s}& \\
\mathbf{0}_{1\times d_{M'(i-1)+1}^s} & \mathbf{v}_{M'(i-1)+2}(t)&\cdots &\mathbf{0}_{1\times d_{M'i}^s} \\
\vdots & \vdots & \ddots & \vdots\\
\mathbf{0}_{1\times d_{M'(i-1)+1}^s} & \mathbf{0}_{1\times d_{M'(i-1)+2}^s} &\cdots &\mathbf{v}_{M'i}(t)\\
\end{array}\right]
\end{align*}
and $\mathbf{v}_{i}(t)=\left[\begin{array}{ccccccc}v_{i,1}(t) & v_{i,2}(t)& \cdots & v_{i,d_{i}^s}(t)\end{array}\right].$ Note that the number of column vectors in $\mathbf{{\bar{V}}}'_i$ is given by $c_i=M'R(n+1)^p$ for $i\in\{1,2,\ldots, K_1\}$, $c_{i}=(R+1-K_1M')R(n+1)^p+((K_1+1)M'-R+1)Rn^p$ for $i=K_1+1$, and $c_i=M'Rn^p$ for $i\in\{K_1+2,K_1+3,\ldots, K\}$, i.e.,
\begin{align*}
c_i= \left\{\begin{array}{ll} M'R(n+1)^p & \textrm{if $i\leq K_1 $}, \\
(R+1-K_1M')R(n+1)^p+((K_1+1)M'-R+1)Rn^p & \textrm{if $i= K_1+1$}, \\
M'Rn^p & \textrm{if $i\geq K_1+2$,}
\end{array} \right.
\end{align*}
where $K_1=\left\lfloor \frac{R+1}{M'}\right\rfloor$. In addition, we set the digital precoder of transmitter $i$ over the extended channel as
\begin{align*}
\mathbf{\bar{V}}_i=\left[\begin{array}{ccccccc} \mathbf{I}_{d_i}\\ \mathbf{0}_{(c_i-d_i)\times d_i} \end{array}\right],
\end{align*}
where $d_i=\min\{MT,NT,c_i\}$. Observe that we can choose first $d_i$ column vectors of $\mathbf{\bar{V}'}_i$ out of $c_i$ vectors by multiplying $\mathbf{\bar{V}'}_i$ and $\mathbf{\bar{V}}_i$ as $\mathbf{\bar{V}'}_i\mathbf{\bar{V}}_i$. Therefore, from the results of Lemma 4, we can see that the dimension of the signal space spanned by the desired signal vectors at receiver $i$ is equal to $d_i\leq c_i$ and the dimension of the signal space spanned by the interference signal vectors at receiver $i$ is less than or equal to $RM'T-c_i$ out of the $RM'T$ dimensional space. Hence, we can null out all the interferences at receiver $i$ via zero forcing beam-forming $\mathbf{\bar{U}'}_i$ over the extended channel, by setting $\mathbf{\bar{U}'}_i$ as the $RM'T\times d_i$ matrix such that $\mathbf{\bar{U}}_i'^{\ast}\mathbf{\bar{H}}_{ij}\mathbf{\bar{V}'}_j\mathbf{\bar{V}}_j=\mathbf{0}$ $\forall i\neq j$ and  $\text{rank}\left(\mathbf{\bar{U}}_i'^{\ast}\mathbf{\bar{H}}_{ii}\mathbf{\bar{V}'}_i\mathbf{\bar{V}}_i\right)=d_i$ $\forall i=1,2,\ldots K$. Furthermore, by setting the receive digital beam-forming matrix of receiver $i$ over the extended channel, $\mathbf{\bar{U}}_i$, as $\mathbf{\bar{U}}_i=\mathbf{I}_{d_i}$, each user $i$ can achieve $d_i$ DoF over the extended channel.

Finally, the achievable sum DoF is given by
\begin{align*}
\Gamma=\frac{1}{T}\sum_{i=1}^{K}d_i&=K_1\min\left\{M,N,\sup_n\frac{M'R(n+1)^p}{(R+1)(n+1)^p}\right\}\\
&+\min\left\{M,N,\sup_n\frac{(R+1-K_1M')R(n+1)^p+((K_1+1)M'-R+1)Rn^p}{(R+1)(n+1)^p}\right\}\\
&+(K-K_1-1)\min\left\{M,N,\sup_n\frac{M'Rn^p}{(R+1)(n+1)^p}\right\}\\
&=K\min\left\{\frac{RM'}{R+1}, M,N\right\},
\end{align*}
which completes the proof of the achievability of Theorem 2.


\subsection{Converse}
\subsubsection{$K\leq R$} Recall that from Lemma 1, the DoF of the $(M,M')\times (N,N')$ PTP channel is equal to $\min\{M,N\}$. Therefore, for the $K$-user $(M,M')\times (N,N')$ interference channel, the sum DoF cannot be more than $K\min\{M,N\}$, i.e., $\Gamma \leq K\min\{M,N\}$.
\subsubsection{$K> R$} Let $d_i$ denote the DoF for each user $i$. Similar in~\cite{Tiangao:10,Viveck1:09}, we first focus on an upper bound on $d_1+d_2+\ldots+d_{R+1}$. We eliminate all the messages except $W_1,W_2,\ldots,W_{R+1}$, which does not decrease $d_1+d_2+\ldots+d_{R+1}$ and results in the $R+1$-user $(M,M')\times (N,N')$ interference channel. Now we allow full cooperation among transmitters $1,2,\ldots, R$ and among receivers $1,2,\ldots, R$ to form the two-user $(M_i,M_i')\times (N_i,N_i')$ interference channel, where $M_1=RM$, $M'_1=RM'$, $M_2=M$, $M'_2=M'$, $N_1=RN$, $N'_1=RN'$, $N_2=N$, and $N'_2=N'$. Then, from the result of Theorem 1, the sum DoF of this channel is given by
\begin{align}
&d_1+\ldots+d_{R+1}\nonumber \\
&=\min\{(R+1)M,(R+1)N,RM+N,RN+M,\max\{RM',N'\}, \max\{RN',M'\}\}\nonumber \\
&=\min\{(R+1)M,(R+1)N,\max\{M',N'\}\}.
\end{align}
Since allowing full cooperation among some transmitters and among some receivers does not reduce the capacity region, (13) is also an upper bound for the original channel. Due to the symmetry, by picking any $R+1$ users out of $K$ users, we have the following upper bound for the original channel:
\begin{align}
d_{i_1}+\ldots+d_{i_{R+1}}&\leq \min\{(R+1)M,(R+1)N,\max\{M',N'\}\},
\end{align}
for all $i_1,i_2,\ldots,i_{R+1}\in\{1,2,\ldots,K\}$ with $i_1\neq i_2\neq \ldots\neq i_{R+1}$. Hence, summing up all such bounds, we finally have
\begin{align*}
\Gamma\leq K\min\left\{M,N,\frac{\max\{M',N'\}}{R+1}\right\}.
\end{align*}

\section{Conclusion}
In this paper, we studied the sum DoF of the $K$-user MIMO interference channels where each user is equipped with a larger number of antennas than the number of RF chains. For the two-user case, the sum DoF was completely characterized for arbitrary numbers of antennas and RF chains. For the $K$-user case ($K\geq 3$), the achievable DoF was derived under the symmetric antenna configuration. It is shown that our achievable scheme is optimal in achieving the sum DoF of the $K$ user hybrid beam-forming systems if the ratio $\frac{\max\{M',N'\}}{\min\{M',N'\}}$ is equal to an integer, where $M'$ and $N'$ denote the number of antennas at each transmitter and receiver, respectively.

Our work has revealed that hybrid beam-forming can provide a significant gain by nulling out interferences between users, and the gain dramatically increases as SNR increases. Moreover, interestingly, even the sum DoF performance of the interference-free channel can be achieved if we add enough number of antennas at either transmitter or receiver side only. Therefore, the results of this paper imply that employing hybrid beam-forming can be an attractive solution for enhancing the capacity of interference-limited networks.

\bibliographystyle{IEEEtran}
\bibliography{IEEEabrv,References}

\begin{thebibliography}{10}
\providecommand{\url}[1]{#1}
\csname url@samestyle\endcsname
\providecommand{\newblock}{\relax}
\providecommand{\bibinfo}[2]{#2}
\providecommand{\BIBentrySTDinterwordspacing}{\spaceskip=0pt\relax}
\providecommand{\BIBentryALTinterwordstretchfactor}{4}
\providecommand{\BIBentryALTinterwordspacing}{\spaceskip=\fontdimen2\font plus
\BIBentryALTinterwordstretchfactor\fontdimen3\font minus
  \fontdimen4\font\relax}
\providecommand{\BIBforeignlanguage}[2]{{%
\expandafter\ifx\csname l@#1\endcsname\relax
\typeout{** WARNING: IEEEtran.bst: No hyphenation pattern has been}%
\typeout{** loaded for the language `#1'. Using the pattern for}%
\typeout{** the default language instead.}%
\else
\language=\csname l@#1\endcsname
\fi
#2}}
\providecommand{\BIBdecl}{\relax}
\BIBdecl

\bibitem{CISCO13}
{CISCO}, ``Cisco visual networking index: {G}lobal mobile data traffic forecast
  update, 2012--2017,'' \emph{White paper}, Feb. 2013.

\bibitem{Marzetta:TWC10}
T.~L. Marzetta, ``Noncooperative cellular wireless with unlimited numbers of
  base station antennas,'' \emph{{IEEE} Trans. Wireless Commun.}, vol.~9,
  no.~11, pp. 3590--3600, Nov. 2010.

\bibitem{LuLiSwindlehurstAshikhminZhang:JSTSP14}
L.~Lu, G.~Y. Li, A.~L. Swindlehurst, A.~Ashikhmin, and R.~Zhang, ``An overview
  of massive {MIMO}: {B}enefits and challenges,'' \emph{IEEE J. Sel. Top.
  Signal Process.}, vol.~8, no.~5, pp. 742--758, Oct. 2014.

\bibitem{ZhangMolischKung:TSP05}
X.~Zhang, A.~F. Molisch, and S.-Y. Kung, ``Variable-phase-shift-based
  {RF}-baseband codesign for {MIMO} antenna selection,'' \emph{{IEEE} Trans.
  Signal Process.}, vol.~53, no.~11, pp. 4091--4103, Nov. 2005.

\bibitem{VenkateswaranVeen:TSP10}
V.~Venkateswaran and A.-J. van~der Veen, ``Analog beamforming in {MIMO}
  communications with phase shift networks and online channel estimation,''
  \emph{{IEEE} Trans. Signal Process.}, vol.~58, no.~8, pp. 4131--4143, Aug.
  2010.

\bibitem{RappaportSunMayzusZhaoAzarWangWongSchulzSamimiGutierrez:Access13}
T.~S. Rappaport, S.~Shun, R.~Mayzus, H.~Zhao, Y.~Azar, K.~Wang, G.~N. Wong,
  J.~K. Schulz, M.~Samimi, and F.~Gutierrez, ``Millimeter wave mobile
  communications for {5G} cellular: {It} will work!'' \emph{IEEE Access},
  vol.~1, pp. 335--349, May 2013.

\bibitem{MarcusPattan:MM05}
M.~Marcus and B.~Pattan, ``Millimeter wave propagation: Spectrum management
  implications,'' \emph{{IEEE} Microwave Mag.}, vol.~6, no.~2, pp. 54--62, Jun.
  2005.

\bibitem{ElAyachRajagopalAbu-SurraPiHeath:TWC14}
O.~E. Ayach, S.~Rajagopal, S.~Abu-Surra, Z.~Pi, and R.~W. {Heath Jr.},
  ``Spatially sparse precoding in millimeter wave {MIMO} systems,''
  \emph{{IEEE} Trans. Wireless Commun.}, vol.~13, no.~3, pp. 1499--1513, Mar.
  2014.

\bibitem{AlkhateebMoGonzalez-PrelcicHeath:CM14}
A.~Alkhateeb, J.~Mo, N.~Gonz\'{a}lez-Prelcic, and R.~W. {Heath Jr.}, ``{MIMO}
  precoding and combining solutions for millimeter-wave systems,'' \emph{{IEEE}
  Commun. Mag.}, vol.~52, no.~12, pp. 122--131, Dec. 2014.

\bibitem{WuChiuLinChang:TWC13}
S.-H. Wu, L.-K. CHiu, K.-Y. Lin, and T.-H. Chang, ``Robust hybrid beamforming
  with phased antenna arrays for downlink {SDMA} in indoor 60 {GHz} channels,''
  \emph{{IEEE} Trans. Wireless Commun.}, vol.~12, no.~9, pp. 4542--4557, Sep.
  2013.

\bibitem{AlkhateebElAyachLeusHeath:JSTP14}
A.~Alkhateeb, O.~E. Ayach, G.~Leus, and R.~W. Heath, ``Channel estimation and
  hybrid precoding for millimeter wave cellular systems,'' \emph{IEEE J. Sel.
  Topics Signal Process.}, vol.~8, no.~5, pp. 831--846, Oct. 2014.

\bibitem{LiangXuDong:WCL14}
L.~Liang, W.~Xu, and X.~Dong, ``Low-complexity hybrid precoding in massive
  multiuser {MIMO} systems,'' \emph{{IEEE} Wireless Commun. Lett.}, vol.~3,
  no.~6, pp. 653--656, Dec. 2014.

\bibitem{Foschini98}
G.~J. Foschini and M.~J. Gans, ``On limits of wireless communications in a
  fading environment when using multiple antennas,'' \emph{Wireless Pres.
  Commun.}, no.~6, pp. 311--335, Mar. 1998.

\bibitem{Telatar99}
E.~Telatar, ``Capacity of multiple-antenna {G}aussian channels,'' \emph{Europ.
  Trans. Telecomm. (ETT)}, vol.~10, no.~6, pp. 585--596, Nov. 1999.

\bibitem{Jafar07}
S.~A. Jafar and M.~J. Fakhereddin, ``Degrees of freedom for the {MIMO}
  interference channel,'' \emph{{IEEE} Trans. Inf. Theory}, vol.~53, no.~7, pp.
  2637--2642, Jul. 2007.

\bibitem{Cadambe107}
V.~R. Cadambe and S.~A. Jafar, ``Interference alignment and degrees of freedom
  for the {$K$}-user interference channel,'' \emph{{IEEE} Trans. Inf. Theory},
  vol.~54, no.~8, pp. 3425--3441, Aug. 2008.

\bibitem{Viveck1:09}
------, ``Degrees of freedom of wireless networks with relays, feedback,
  cooperation, and full duplex operation,'' \emph{{IEEE} Trans. Inf. Theory},
  vol.~55, no.~5, pp. 2334--2344, May 2009.

\bibitem{Viveck2:09}
------, ``Interference alignment and the degrees of freedom of wireless {$X$}
  networks,'' \emph{{IEEE} Trans. Inf. Theory}, vol.~55, no.~9, pp. 3893--3908,
  Sep. 2009.

\bibitem{Motahari:091}
A.~Motahari, S.~O. Gharan, and A.~Khandani, ``Real interference alignment:
  Exploiting the potential of single antenna systems,'' \emph{{IEEE} Trans.
  Inf. Theory}, vol.~60, no.~8, pp. 4799--4810, Aug. 2014.

\bibitem{Suh08}
C.~Suh and D.~Tse, ``Interference alignment for cellular networks,'' in
  \emph{Proc. 46th Annu. Allerton Conf. Communication, Control, and Computing},
  Monticello, IL, Sep. 2008, pp. 1037--1044.

\bibitem{Suh:11}
------, ``Downlink interference alignment,'' \emph{{IEEE} Trans. Commun.},
  vol.~59, no.~9, pp. 2616--2626, Sep. 2011.

\bibitem{Akbar10}
A.~Ghasemi, A.~S. Motahari, and A.~K. Khandani, ``Interference alignment for
  the {$K$} user {MIMO} interference channel,'' in \emph{Proc. {IEEE} Int.
  Symp. Information Theory}, Austin, TX, Jun. 2010, pp. 360--364.

\bibitem{Tiangao:10}
T.~Gou and S.~A. Jafar, ``Degrees of freedom of the {$K$} user {$M\times N$}
  {MIMO} interference channel,'' \emph{{IEEE} Trans. Inf. Theory}, vol.~56,
  no.~12, pp. 6040--6057, Dec. 2010.

\bibitem{Jeon2:11}
S.-W. Jeon, S.-Y. Chung, and S.~A. Jafar, ``Degrees of freedom region of a
  class of multisource {G}aussian relay networks,'' \emph{{IEEE} Trans. Inf.
  Theory}, vol.~57, no.~5, pp. 3032--3044, May 2011.

\bibitem{Tiangao:12}
T.~Gou, S.~A. Jafar, C.~Wang, S.-W. Jeon, and S.-Y. Chung, ``Aligned
  interference neutralization and the degrees of freedom of the
  $2\times2\times2$ interference channel,'' \emph{{IEEE} Trans. Inf. Theory},
  vol.~58, no.~7, pp. 4381--4395, Jul. 2012.

\bibitem{Annapureddy:11}
V.~S. Annapureddy, A.~{El Gamal}, and V.~V. Veeravalli, ``Degrees of freedom of
  interference channels with {CoMP} transmission and reception,'' \emph{{IEEE}
  Trans. Inf. Theory}, vol.~58, no.~9, pp. 5740--5760, Sep. 2012.

\bibitem{Ke:12}
L.~Ke, A.~Ramamoorthy, Z.~Wang, and H.~Yin, ``Degrees of freedom region for an
  interference network with general message demands,'' \emph{{IEEE} Trans. Inf.
  Theory}, vol.~58, no.~6, pp. 3787--3797, Jun. 2012.

\bibitem{Jeon4:12}
S.-W. Jeon and M.~Gastpar, ``A survey on interference networks: Interference
  alignment and neutralization,'' \emph{Entropy}, vol.~14, no.~10, pp.
  1842--1863, Sep. 2012.

\bibitem{ChaeJ1}
S.~H. Chae, S.~W. Choi, and S.-Y. Chung, ``On the multiplexing gain of
  {$K$}-user line-of-sight interference channels,'' \emph{{IEEE} Trans.
  Commun.}, vol.~29, no.~10, pp. 2905--2915, Oct. 2011.

\bibitem{Chae11}
S.~H. Chae and S.-Y. Chung, ``On the degrees of freedom of rank deficient
  interference channels,'' in \emph{Proc. {IEEE} Int. Symp. Information
  Theory}, Saint Petersburg, Russia, Jul.-Aug. 2011, pp. 1367--1371.

\bibitem{Krishnamurthy12}
S.~R. Krishnamurthy and S.~A. Jafar, ``Degrees of freedom of 2-user and 3-user
  rank-deficient {MIMO} interference channels,'' in \emph{Proc. 2012 IEEE
  Global Telecommunications Conference (GLOBECOM)}, Anaheim, USA, Dec. 2012,
  pp. 2462--2467.

\bibitem{LNIT}
A.~{El Gamal} and Y.-H. Kim, \emph{Network Information Theory}.\hskip 1em plus
  0.5em minus 0.4em\relax Cambridge, 2010.

\bibitem{DIA}
K.~Gomadam, V.~R. Cadambe, and S.~A. Jafar, ``A distributed numerical approach
  to interference alignment and applications to wireless interference
  networks,'' \emph{{IEEE} Trans. Inf. Theory}, vol.~57, no.~6, pp. 3309--3321,
  Jun. 2011.

\end{thebibliography}

\end{document}